\documentclass[alpha-refs]{wiley-article}
\usepackage[utf8]{inputenc}

\usepackage{amsmath}
\usepackage{booktabs}
\usepackage[linesnumbered,lined,vlined]{algorithm2e}

\newcommand{\opd}{\mathrm{d}}
\newcommand{\bP}[1]{\mathbf{P}\big(#1\big)}

\title{A competing risks interpretation of Hawkes processes}
  \author[1]{Maximilian Aigner}
  \author[2]{Valérie Chavez-Demoulin}
  \affil[1,2]{Department of Operations, University of Lausanne, Lausanne, Vaud, 1015, Switzerland}
  
\corraddress{Maximilian Aigner, Department of Operations, University of Lausanne, Lausanne, Vaud, 1015, Switzerland}
\corremail{maximilian.aigner@unil.ch}
\begin{document}
\maketitle


\begin{abstract}
We give a construction of the Hawkes process as a piecewise competing risks model. We argue that the most natural interpretation of the self-excitation kernel is the hazard function of a defective random variable. This establishes a link between desired qualitative features of the process and a parametric form for the kernel, which we illustrate using examples from the literature. Two families of cure rate models taken from the survival analysis literature are proposed as new models for the self-excitation kernel. Finally, we show that the competing risks viewpoint leads to a general simulation algorithm which avoids inverting the compensator of the point process or performing an accept-reject step, and is therefore fast and quite general.
\end{abstract}



\section{Introduction}

The Hawkes process \citep{hawkes_point_1971} is a model for self-exciting events occurring in continuous time. It is a point process $N(t)$ on the real line, parametrised by the conditional intensity function \citep{hawkes_point_1971, daley_introduction_2003}
\begin{equation}
\lambda(t) = \eta + \int_{-\infty}^t h(t - s)\, N(\opd s).
\label{eq:cif}
\end{equation}
The \emph{baseline intensity} parameter $\eta > 0$ represents the rate of background events, while the \emph{excitation kernel} $h(t)$, a positive function defined on the positive real line, models the influence of previous events on the conditional intensity. We assume 
\begin{equation}
\gamma = \int_0^\infty h(t) \opd t < 1,
\end{equation} which ensures that the process defined by \eqref{eq:cif} exists \citep[§6.3c]{daley_introduction_2003}.

The Hawkes process can be viewed \citep{hawkes_cluster_1974} as a generalised branching process, where exogenous events arrive at the times of a Poisson process with constant rate $\eta$. Each exogenous event $t_i$ then triggers offspring events according to an inhomogeneous Poisson process with rate $h(\cdot - t_i)$. Each offspring event in turn generates further events, and so on. One realisation of the Hawkes process is the union of all the events of all generations. 

Because the excitation kernel locally increases the conditional rate of new points according to \eqref{eq:cif}, the Hawkes process naturally models self-exciting processes, such as epidemics \citep{rizoiu_sir-hawkes_2018}, earthquakes \citep{ogata_statistical_1988, marsan_extending_2008}, or crime \citep{lewis_nonparametric_2011}; see \cite{hawkes_hawkes_2018} and \cite{reinhart_review_2018} for reviews of applications. 

Self-exciting behaviour is best observed at high temporal precision as it can be masked by aggregation at too great a scale \citep{cheysson_strong_2020}. Therefore, the use of continuous-time models such as the Hawkes process has increased in line with the availability of high-precision temporal data from modern technologies such as always-connected devices. Recent applications of the Hawkes process thus include analysis of user activity on social networks \citep{salehi_learning_2019}, fraud detection \citep{dutta_hawkeseye_2020}, or application to high-frequency trading \citep{bacry_estimation_2016, bacry_first-_2016, bacry_hawkes_2015}. The same argument can be made regarding spatial precision \citep{adepeju_investigating_2017, mohler_self-exciting_2011}.

We note that the Hawkes process has been theoretically extended in multiple ways, such as by considering multiple point processes and their cross-excitation \citep{hawkes_point_1971},
by allowing the baseline rate to vary in time \citep{wheatley_hawkes_2016}, 
or by linking the conditional intensity to the linear predictor \eqref{eq:cif} through a nonlinear function \citep{bremaud_stability_1996, costa_renewal_2018}. For the sake of clarity and to underline our argument, we restrict ourselves in this work to the simple, univariate case specified above. The conclusions, however, apply to all of the extensions mentioned except the nonlinear Hawkes process of \cite{bremaud_stability_1996}.

\subsection{Parameterisations of the Hawkes process}
Aside from the scalar parameter $\eta$, the estimation of a Hawkes process model depends on either providing a parametric form for $h(t)$ or on estimating it non-parametrically. The only requirements imposed on $h(t)$ are positivity and the stability requirement $\gamma < 1$. This general setting means that the possibilities for both parametric or non-parametric specifications of $h(t)$ are vast.

A great variety of parametric forms have been proposed for the function $h(t)$, motivated either by domain knowledge or computational efficiency. An important example is the famed Omori law of aftershocks \citep{ouillon_magnitude-dependent_2005} in seismology, which has led to the ETAS kernel \citep{ogata_statistical_1988, ogata_space-time_1998}
$$
h(t) = \frac{K}{(t+c)^{p}}, \quad(c \geq 0, p \geq 0)
$$ according to which the rate of offspring events decreases as a power law of the time elapsed. 
Another common choice is the exponential-type kernel $h(t) = \alpha e^{-\beta t}$, whose advantages are computational tractability \citep{hawkes_point_1971, hawkes_spectra_1971, ogata_asymptotic_1978, ogata_linear_1982} 
and the Markovian nature of the corresponding Hawkes process \citep{liniger_multivariate_2009}.

The function $h(t)$ is usually interpreted as representing a time-changing ``influence'' \citep{xu_dirichlet_2017} or ``excitation effect'' \citep{li_fm-hawkes_2017}. The shape of $h(t)$ is chosen to match qualitative features such as the speed of decay to 0 or the refractory period $a \geq 0$ such that $h(t) = 0$ for $t \leq a$, representing a delayed effect \citep{eichler_graphical_2017}. Still, in the absence of further constraints, the number of possible parametric specifications for $h(t)$ remains too large.

Non- or semi-parametric estimation of $h(t)$ has been explored extensively, using spline bases \citep{zhou_learning_2013-1, reynaud-bouret_adaptive_2010, carstensen_multivariate_2010}, histogram and kernel density estimation \citep{lewis_nonparametric_2011, kirchner_estimation_2017, eichler_graphical_2017}, and Gaussian process regression \citep{dowling_non-parametric_2020,zhang_variational_2020}. The many different approaches reflect the freedom of choice of the parameter $h(t)$, and yield satisfactory results given sufficient data. However, these various methods can cause conceptual and computational difficulties, and it is possible that more parsimonious models could do equally well.

In this paper, we aim to address the question of how to qualitatively assess a choice of $h(t)$ for a given set of data, and what interpretation to give to this function. We will argue that based on conceptual and computational grounds, setting $h(t)$ to be a \emph{hazard} function is the most natural choice. 
Section \ref{sec:arguments} gives several arguments to this effect, based on theory as well as applications. Section \ref{sec:cure_rate} interprets existing models and proposes two ways to generate appropriate models for $h(t)$ as a hazard function. Furthermore, in Section \ref{sec:simulation} we use the hazard function interpretation to develop a fast simulation algorithm. Avenues for further research are described in the concluding Section \ref{sec:conclusion}.


\section{Interpretation of the excitation function}
\label{sec:arguments}

We recall the general definition of a conditional intensity for a point process $N$: it is a stochastic process $\lambda(t)$ such that
\begin{equation}
\mathbb{E}\left[ N(a, b) \,\bigl\vert \, \mathcal{H}_a \right] = \mathbb{E}\left[ \int_a^b \lambda(t) \opd t \, \bigl\vert \, \mathcal{H}_a\right],
\label{eq:def_cif}
\end{equation}
where $\mathcal{H}_a = \sigma(\left\{N(0, s) \mid 0 < s < a\right\})$, a sigma algebra, is the history of the process at time $a$, including the previous event times.

In the specific case of point processes on $\mathbb{R}_+$, $\lambda(t)$ can be constructed piecewise as the hazard rate between two events of the process \citep{daley_introduction_2003}: consider the situation where $i-1$ times of the process have been observed and denote them by $t_1, \ldots, t_{i-1}$; for $t > t_{i-1}$, define $f_i(t\, | \mathcal{H}_{t_{i-1}})$ as the conditional density of the $i$-th event time given the history. Then $\lambda(t)$ is defined in the $i$-th interval $t_{i-1} < t < t_i$ as the hazard rate,
$$
\lambda(t) = h_i(t |\mathcal{H}_{t_{i-1}}) = \frac{f_i(t | \mathcal{H}_{t_{i-1}})}{1-F_i(t | \mathcal{H}_{t_{i-1}})}.
$$
While this definition does not readily extend to processes on more general spaces, it is more intuitive in our setting.

The following sub-sections give arguments for the interpretation of $h$ as a hazard function.

\subsection{Competing risk models}
\label{sec:arg_competingrisks}

First, the Hawkes process can be seen as an application of the competing risk framework used in survival analysis and reliability studies. This framework addresses the situation when an event time (death, failure, etc) can arise through one of $J$ competing causes, and the cause is observed only at the event time. In this situation, a classical viewpoint (the \emph{latent class} model) is to define $J$ \emph{latent} times associated with each event, and to assume that the observed failure time corresponds to the minimum of the latent times (first failure). Then the overall hazard $\lambda(t)$ of any individual can be decomposed additively into $J$ cause-specific hazards, measuring the instantaneous probability of death from each cause. For more details see e.g.~\cite{kalbfleisch_statistical_2002}. 

In the Hawkes process context, an event at time $t$ is caused by one of the $N(t)$ preceding events or the background, and may be the cause (or parent) of later events. 
While the latent class model has drawbacks, particularly relating to estimation of the cause-specific hazards \citep{crowder_identifiability_1994}, in this case the cause-specific hazards $h(t - t_i)$ are all time-shifted versions of the same hazard function $h$, and we are not interested in the cause-specific hazards themselves.
Moreover, the cause of each event is not or only partially observed in a Hawkes process context, which leads to the missing-data EM approach detailed in the next section, and the number of causes is no longer fixed but varies with time. We note in particular that the additive decomposition of risks applies only to the hazard function, not to the density function \cite[ch.8]{kalbfleisch_statistical_2002}. This motivates our interpretation over, say, the interpretation that $h(t) = \gamma f(t)$ with $f$ a density function and $\gamma < 1$, as proposed by \cite{xu_dirichlet_2017} and others.

\subsection{Complete-data likelihood}

The Hawkes process \emph{complete-data} likelihood, which is obtained by conditioning on the branching process structure of the events as well as their times, can be expressed as a survival model likelihood. Indeed, let $z_{ji} = 1$ when $i$ is an offspring event of $j$ and $0$ otherwise, with the convention that $z_{0i} = 1$ refers to background events. Then $\mathbf{z} = (z_{ji})_{j,i=1}^n$ contains all the branching structure of the observations $t_1, \ldots, t_n$. For data observed in the interval $[0, T]$, say, the complete data log-likelihood of some parameter vector $\theta$ is then \citep{lewis_nonparametric_2011}
$$
\mathrm{cl}(\theta \big\lvert \mathbf{z}, \mathbf{t}) = \sum_{i=1}^n z_{0i} \log\eta_\theta - \eta_\theta T + \sum_{i=2}^n \left[\sum_{j=1}^{i-1} z_{ji} \log h_\theta(t_i - t_j) - \int_{t_i}^T h_\theta(t-t_j) \opd t\right]
$$

The first two terms are the likelihood of a homogeneous Poisson process with intensity $\eta_\theta$, applied to the background events. The second part can be rewritten in terms of the density $f$ and survival function $S$ associated with $h$ as

\begin{align}
\log \left[ \prod_{i=2}^n \prod_{j=1}^{i-1} h_\theta(t_i - t_j)^{z_{ji}} e^{-\int_{t_j}^{t_i} h_\theta(t-t_j) \opd t}\right] &= \log \left[ \prod_{i=2}^n \prod_{j=1}^{i-1} f_\theta(t_i - t_j)^{z_{ji}} S_\theta(t_i - t_j)^{1-z_{ji}} \right]
\label{eq:likelihood_prod}
\end{align}

The expression on the right-hand side of \eqref{eq:likelihood_prod} is the likelihood of survival times $t_i - t_j$, with hazard function $h(t)$, which are observed when $z_{ji} = 1$, and censored when $z_{ji} = 0$. In other words, only the causal time is observed and all the others are censored; this property is shared with the competing risk model \citep{kleinbaum_competing_2012}. 
This likelihood no longer holds if $h$ is now the density of offspring times rather than the hazard function.

The complete-data log-likelihood is essential to an EM-type estimation for Hawkes processes \citep{mohler_self-exciting_2011}; optimising the expected complete-data log-likelihood makes up the M-step of such an algorithm. The E-step is the subject of the next section.

\subsection{Stochastic declustering}

The \emph{stochastic declustering} technique \citep{marsan_extending_2008, veen_estimation_2008, mohler_self-exciting_2011} is naturally derived from a hazard function interpretation of $h(t)$. This technique estimates the branching process structure of a set of events and has been used for classifying earthquakes as initial or aftershocks \citep{marsan_extending_2008}.
It provides the ``missing data'' required in the previous step, namely the branching process structure represented by the $z_{ji}$ causal indicators.
Given the data $t_1, \ldots, t_n$ we may derive by Bayes' rule
\begin{align*}
P(z_{ji} = 1 \big\lvert t_1, \ldots, t_n) &= \frac{f_i(t_i | z_{ji}=1)}{f_i(t | z_{0i} = 1) + \sum_{k=1}^n f_i(t | z_{ki} = 1)} \\
&= \frac{h_j(t_i) S(t_i)}{\sum_{k=0}^n h_k(t_i) S(t_i)}
\end{align*}
where $h_j$ is the hazard of $t_i$ when caused by $t_j$ and $S(t_i)$ is the overall survival function of $t_i$. This can again be written as
\begin{equation}
P(z_{ji} = 1 \big\lvert t_1, \ldots, t_n) = \frac{h(t_i - t_j)}{\eta + \sum_{k=1}^n h(t_i - t_k)},
\label{eq:e-step}
\end{equation}
which is the standard formula found in e.g.~ \cite{mohler_self-exciting_2011}. Note the background term $\eta$ in the denominator of \eqref{eq:e-step}, which cannot be interpreted as a density.

\subsection{Inverse compensator of the Hawkes process}

Finally, consider the task of simulating a sample $(T_i)_{i=1}^n$ from a point process with stochastic intensity $\lambda(t)$, or more generally with a compensator $\Lambda(t)$ which when absolutely continuous is represented as $\Lambda(t) = \int_0^t \lambda(s) \opd s$. The random time change theorem (see e.g. \cite{bremaud_point_1981,brown_time-rescaling_2002}) implies that the rescaled event times of $N$, $\Lambda(T_i)$, are stochastically equivalent to the event times of a Poisson process with unit intensity. Working backwards, if the compensator is not too hard to invert, a sample from $N$ in the interval $[0, T]$ can be simulated by the equation
\begin{equation}
T_n \overset{d}{=} \Lambda^{-1}\left(\sum_{i=1}^n E_i \right)
\label{eq:invert_comp_sim}
\end{equation}
where $E_i$ are i.i.d.~ standard exponential variates.
The Hawkes process compensator has the form
$$
\Lambda(t) = \eta t + \sum_{i=1}^{N(t)} \int_{t_i}^t h(s-t_i) \opd s
$$
and the sum makes this difficult to invert as a function of $t$. If we choose for $h(t)$ a hazard function, however, and define the cumulative hazard $H(t) = \int_0^t h(s) \opd s$, then the compensator of $N$ may be written as 
$$
\Lambda(t) = \eta t + \sum_{i=1}^{N(t)} H(t - t_i)
$$
and then the overall survival function is
\begin{equation}
S(t) = e^{-\Lambda(t)} = e^{-\eta t} \prod_{i=1}^{N(t)} e^{-H(t-t_i)} =: S_0(t) \prod_{i=1}^{N(t)} S_i(t).
\label{eq:overall_survival}
\end{equation}
Note that equation \eqref{eq:overall_survival} is also the survival function of the minimum of $N(t) + 1$ times, associated with the survival functions $S_1(t), \ldots, S_{N(t)}(t), S_0(t)$ respectively. Thus, if each compensator associated in the sum is easily invertible, then the time of the next event can be simulated by inverting the individual compensators and taking the minimum. %
This can still be difficult in practice, however; for the models considered in this paper and those drawn from survival analysis, simulation is usually straightforward. The details of a simulation algorithm using the method above are given in Section \ref{sec:simulation}.

We believe these arguments are sufficient to establish the interpretation of $h$ as a hazard function. In the next two sections, we describe consequences of this interpretation for modelling and simulation of Hawkes processes.


\section{Cure rate models for the Hawkes process}
\label{sec:cure_rate}

As mentioned previously, a large number of choices for $h$ exist in the literature. We review these and compare them from the perspective of the associated distributions. 

Note that if $h$ is a hazard function, then the stability criterion $\int_0^{\infty} h(t) \opd t < 1$ bounds the cumulative hazard, that is, the total expected number of events that are offspring of any event. Because $\lim_{t\to\infty} F(t) = \lim_{t\to\infty} 1 - e^{\int_0^{\infty} h(t) \opd t} < 1$, the distribution associated with $h$ must be \emph{defective} in the sense that it takes value $+\infty$ with nonzero probability. 
In our setting, we have
$$
\gamma = \int_0^\infty h(t) \opd t = -\log \left[1 - F(+\infty)\right] < 1
$$
and therefore the probability with which the variable is infinite, and the corresponding event does not happen, must exceed $1 - 1/e \simeq 0.63$.
From a survival analysis perspective, a defective event-time distribution implies that a certain part of the population does not experience the event of interest, and models for this phenomenon are called \emph{cure rate} or \emph{cure fraction} models. 

\subsection{Existing models}

\begin{table}[ht]
    \caption{Commonly used excitation functions and their corresponding distributions. \label{tab:hazards_literature}}
    \centering
    \bigskip
    
    \makebox[\textwidth][c]{
    \begin{tabular}{lllr}
    \toprule[2pt]
    Excitation & $h(t)$ & & Distribution\\
    \cmidrule[1pt](rl){1-4}
    Exponential & $\alpha e^{-\beta t}$ & $a\geq 0, b > 0$ & $\text{Gompertz}(\alpha, -\beta)$ \\
        \addlinespace[0.5em]
    Omori Law/ETAS & $\dfrac{K}{(t + c)^{p}}$ & $K\geq 0, c > 0, p > 0$ & Weibull if $p \leq 1$, ``Omori'' if $p > 1$\\
    \addlinespace[0.5em]

    Nonparametric & $\sum_{k=1}^{N_h} \alpha_k \mathbf{1}_{\left[(k-1)\Delta, k\Delta\right)}(t)$ & $\alpha_k \geq 0, \Delta > 0, N_h \in \mathbb{N}$ & Piecewise Exponential \\
        \addlinespace[0.5em]
    \bottomrule[2pt] 
    \addlinespace[0.5em]
    \end{tabular}
    }
\end{table}


Table \ref{tab:hazards_literature} lists some common choices for $h$. The most common choice by far is the exponential hazard function, which we identify as the hazard function of a Gompertz distribution with negative shape parameter. This model has been used in previous work on disease incidence (where it models a fraction of immune subjects) and growth curve modelling, where it models the adoption of new technology. Qualitatively, a Gompertz law models a high initial excitation effect (or mortality) followed by rapid decay. Those times which pass the initial high mortality period are likely to survive indefinitely (i.e.~be equal to $+\infty$). The Gompertz model is a classical model for population with cure fraction \citep{gieser_modelling_1998}.

The Omori law has been mentioned previously as a model for the frequency of aftershocks after an initial earthquake. It is also called the Epidemic-Type After Shock (ETAS) kernel. The hazard function has power-law decay as $t \to \infty$, which has been interpreted as a ``long-memory'' effect \citep{bacry_estimation_2016}. The distribution is defective when $p > 1$, whereas for $0 < p \leq 1$ the hazard corresponds to that of a 3-parameter Weibull distribution. The corresponding waiting-time distribution is similar in some respects to a Gumbel Type II distribution truncated to the positive half-line, but with a negative shape parameter. Since we could not identify another name for this distribution, we will simply refer to it as the Omori distribution.

We note in this regard that the Omori distribution with $p = 1$ is not admissible as a Hawkes excitation, since it does not lead to a stable process: with $\gamma \geq 1$, the number of events within any finite interval is infinite with some nonzero probability; in the seismological context, this would represent a main quake triggering a series of aftershocks which either never dies out, or worse, increases in frequency without bound. Typically, though, the estimates from earthquake catalogues do not exclude the possibility that $p = 1$. One possible answer is that the catalogues themselves are limited in time, creating a censoring effect on the Omori law \citep{hainzl_statistical_2016} that renders the process stable (see Section \ref{sec:tech_def} below). The Omori law is usually motivated as a power-law decay, but it has received criticism practically since its inception \citep{mignan_revisiting_2016-1}. In fact, recent work \citep{hainzl_testing_2017, mignan_modeling_2015} has indicated that the Omori law fails to hold in many cases, and that the tail behaviour of $h(t)$ must be lighter than $t^{-p}$. A promising alternative is the stretched exponential function,
$$
h(t) = t^{\beta - 1} e^{-\lambda t^\beta}
$$
which is justified based on physical considerations of stress-release mechanics \citep{mignan_modeling_2015}. Because $\int_0^{+\infty} h(t) \opd t = (\lambda \beta)^{-1}$, the induced distribution is defective for appropriate values of $\lambda$ and $\beta$. However, we could not identify it as the hazard function of a distribution, though it strongly resembles the Weibull Extension of \cite{xie_modified_2002}.

Finally, though the above parametric distributions are quite well known, the correspondence established in Section \ref{sec:arguments} also extends to a ``nonparametric'' form of $h$ which has been used in the absence of a specific parametric form for the excitation functions \citep{eichler_graphical_2017,kirchner_estimation_2017}. This consists in setting 
$$
h(t) = \sum_{k=1}^{N_h} \alpha_k \mathbf{1}_{\left[(k-1)\Delta, k\Delta\right)}(t)
$$
with fixed $N_h, \Delta$ and estimating the parameters $\alpha_k$. As a specification for the hazard function, this type of model has been used extensively in survival analysis \citep{laird_covariance_1981, weiss_piecewise-constant_2017, balakrishnan_piecewise_2016, bouaziz_l0_2017}. 
The corresponding distribution is a composite (piecewise) exponential distribution with rates given by the $\alpha_k$. Simulation from this law is straightforward, for example using the well-known alias method for sampling from a discrete distribution.

\subsection{Technically defective distributions}
\label{sec:tech_def}

As mentioned previously, we are interested in new parametric models for the hazard function; in light of the previous discussion, these are hazards of defective random variables. In the following, we will distinguish between \emph{natural} defective distributions, which arise from known distributions through the addition of new parameters, and those distributions which are rendered defective through some technical device. Indeed, any distribution traditionally used in survival analysis, such as the Weibull distribution, can be rendered defective by truncation, admixture with a defective distribution, or both.
\begin{enumerate}
    \item Truncation: When the hazard function is set to zero beyond a certain time $\kappa$, yielding the truncated hazard $h(t) \mathbf{1}(t \leq \kappa)$, then the integrated hazard function is flat for $t \geq \kappa$. The distribution therefore contains a defective fraction $1 - F(\kappa)$, and the parameter $\kappa$ can be set to maintain the stability criterion $\int_0^\infty h(t) \opd t < 1$. This mechanism is often plausible in practical settings since the memory of the process is usually not infinite, but rather the process at time $t$ depends only on the events in $[t-\kappa, t)$.
    \item Mixture of immunes: Adopting a mixture model of the form $F^*(t) = \pi F(t) + (1 - \pi) \times 0$ yields a defective fraction $1 - \pi$. This model has been used in cure rate modelling, particularly to relate covariates to the cure probability. In this case, each event either has no offspring (with probability $1-\pi$) or generates offspring in the usual way (with probability $\pi$) according to the hazard function $h(t)$.
\end{enumerate}
To deal with the stability issues of the Omori law, several authors have advocated truncation, though it places an upper limit on the memory of the process. However, in other contexts the ETAS kernel has been chosen because of its long-memory character \citep{bacry_estimation_2016}. In the next section, we describe how models such as the ETAS kernel can be rendered defective through the addition of further parameters.

\subsection{Families of defective distributions}
\label{sec:families}

Natural defective distributions are quite uncommon given that they are not distributions in the classical sense. Two known distributions which are rendered defective when parameters are chosen outside their usual ranges are the Gompertz and Inverse Gaussian distributions \citep{rocha_new_2017-1}. Further distributions can be generated by augmenting existing defective distribution with additional parameters in order to accommodate different shapes of the hazard function; a large class of such generalised (or exponentiated) distributions is generated by the Kumaraswamy transformation \citep{rocha_new_2017-1,cordeiro_kumaraswamy_2010}. Alternatively, we can use the Marshall-Olkin transformation \citep{marshall_new_1997, cordeiro_marshall-olkin_2014} to turn an improper distribution (say, when the parameter value is taken outside its permitted range) into a defective distribution \citep{rocha_new_2017}. We treat each of these in turn.

The Marshall-Olkin family generated by a baseline distribution with hazard $h(t)$ is a family of distributions indexed by $r \geq 0$ and defined by the hazard rate or equivalently by the survival function:
\begin{align}
h_r(t) &= \frac{h(t)}{1 - (1 - r)S(t)} \quad \text{and} \label{eq:mo_def}\\
S_r(t) &= \frac{rS(t)}{1 - (1 - r)S(t)} \nonumber
\end{align}
This results in a defective fraction of
$$
S_r(\infty) = \frac{rS(\infty)}{1-(1-r)S(\infty)},
$$
i.e. the Marshall-Olkin conserves the defective character of $S$ \citep{rocha_new_2017}.
The parameter $r$ can be interpreted as shifting the hazard curve up or down depending on whether $r > 1$ or $r < 1$.

Moreover, if we form an improper distribution such that the survival function $S(t) \geq 1$ and $S(t) \to \infty$ as $t \to \infty$, then the Marshall-Olkin transform of $S$ is a defective distribution \citep{rocha_new_2017}. This yields defective versions of well-known laws such as the Weibull, Extended Weibull, Chen \citep{chen_new_2000}, etc.

Another family of defective distributions can be generated by the Kumaraswamy transformation. For a distribution with survival distribution $S(t)$, the Kumaraswamy-$S$ distribution with parameters $r, u > 0$ is given by \citep{rocha_new_2017-1}
\begin{align}
S_{r,u}(t) &= 1 - \left[1 - S(t)^r\right]^u \text{ and} \label{eq:kum}\\
h_{r,u}(t) &= \frac{ruf(t)F(t)^{r-1}\left[1-F(t)^r\right]^{u-1}}{1 - \left[1 - F(t)^r \right]^u}
\end{align}
The form \eqref{eq:kum} subsumes many so-called \emph{exponentiated} distributions related to $S$. It can be interpreted as modelling the rate of failure of $u$ serial components and $r$ parallel components; see \cite{cordeiro_kumaraswamy_2010}. The defective fraction of the distribution $S_{r,u}$ is given by
$
S_{r,u}(\infty) = 1 - \left[1 - S(\infty)^r\right]^u,
$
and again $S_{r,u}$ is defective only if $S$ is as well \citep{rocha_new_2017-1}.

\subsection{Example}

As an example, we extend the well-known Gompertz distribution (which corresponds to the exponential-type hazard rate) by an additional scale and shape parameter via the Kumaraswamy transformation. As can be seen from Figure \ref{fig:ku_gompertz}, the resulting hazard function can accommodate many different types of excitation, including delayed effects, slow decay and very fast drop-off.
\begin{figure}[ht]
    \centering
    \includegraphics[width=0.8\textwidth]{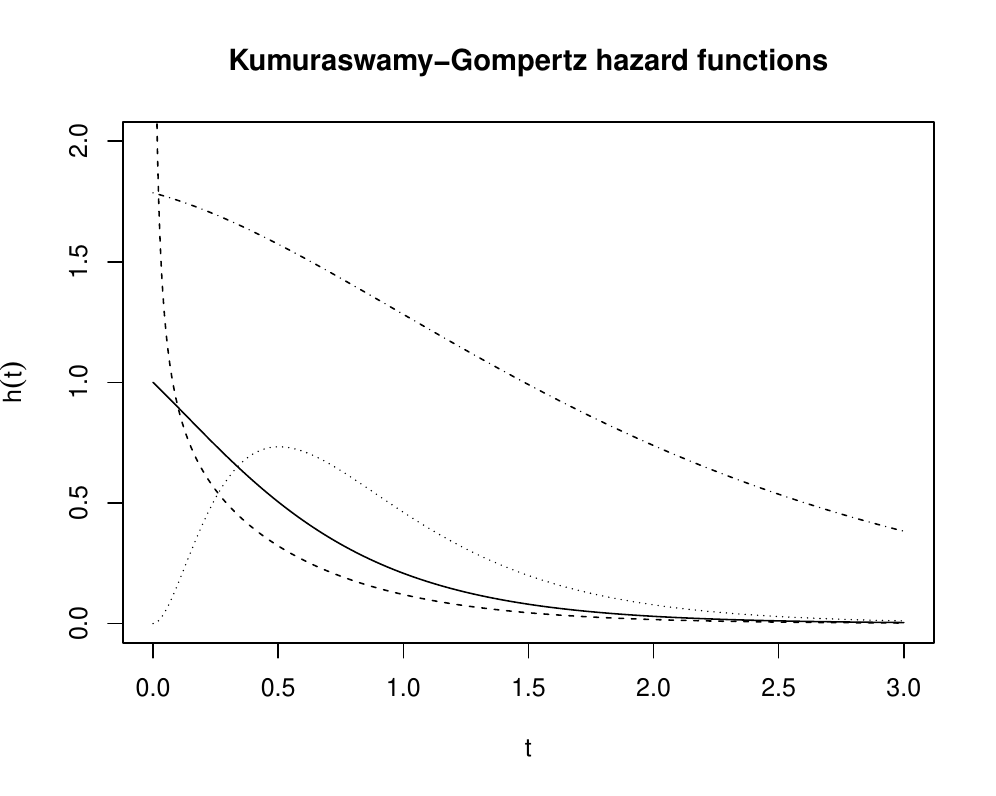}
    \caption{Hazard functions for various Kumaraswamy-Gompertz distributions. The parameters in the form $(\alpha, \beta, r, u)$ are: $(1, 2, 1, 1)$ (solid); $(0.875, 2, 0.5, 0.5)$ (dashed); $(1.09, 2, 3, 7)$ (dotted); $(0.595, 0.75, 1, 3)$ (dot-dashed). These represent different shapes while keeping the defective fraction fixed at $0.5$.}
    \label{fig:ku_gompertz}
\end{figure}

We fit a Hawkes process with this excitation function to a small USGS earthquake catalogue ($n = 2139$) consisting of events in the North American region with intensity greater than $6$, in the period 2000-2020 \citep{usgs_catalog}. The obtained parameters(s.e.) are $\alpha = 0.1098(2\times 10^{-7}), \beta = 0.11(2 \times 10^{-7}), r = 2.407(1.2\times 10^{-7}), u = 2.4075(6.1 \times 10^{-7}), \eta = 0.00714(5.1 \times 10^{-3})$, which indicates that the additional shape parameters $r$ and $u$ are needed to capture the excitation structure. The estimated hazard (or rate of aftershocks) is shown in Figure \ref{fig:fitted_hazard} below. It exhibits a small refractory period, peaks at around 10 days and decreasing rapidly thereafter. We also compute a nonparametric estimate of $h(t)$ by binning the events by $\Delta = 2$ days and counting numbers in each bin, after which an autoregressive model of order 30 is estimated on the binned data. For more details on the nonparametric procedure see \cite{kirchner_estimation_2017}. We defer a full application of the models considered here for various seismological datasets to a later paper.
\begin{figure}[ht]
    \centering
    \includegraphics[width=0.8\textwidth]{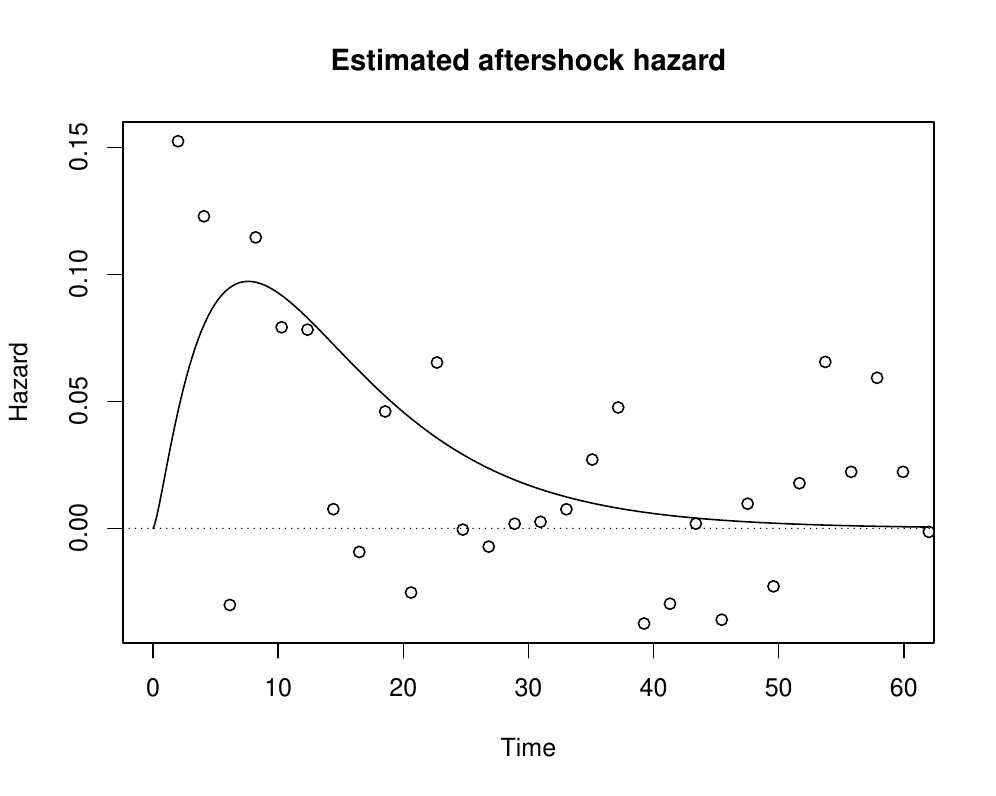}
    \caption{Estimated hazard function on the USGS earthquake catalogue, as a function of time delay in hours. Overlaid points represent a nonparametric estimate of $h$.}
    \label{fig:fitted_hazard}
\end{figure}

While the families discussed in \ref{sec:families} cover many interesting cases, it is clear that still more general families of defective distributions can be obtained by compounding both families, and there may be other such families still unknown. A comprehensive overview is out of the scope of this paper; instead we advocate simply to extend the existing distributions by one or two additional parameters and thereby gain the flexibility necessary for practical applications. 

\section{A construction of the Hawkes process}
\label{sec:simulation}

In this section, we detail an alternative data generating process (DGP) for the Hawkes process, which is amenable to simulation and easy to implement. This section is self-contained and starts from first principles; the result is summarised in Algorithm \ref{alg:sim}. 

Suppose $F$ is the CDF of a positive random variable and $\eta > 0$ is some positive number. Let us construct a point process, say $N$, in the following way. 
\begin{enumerate}
\item Initially, that is, at time $0$, we start a timer with some random, exponentially-distributed duration, say $V^{(0)}_0 \sim \mathrm{Exp}(\eta)$. When the timer runs out at $T_1 := V^{(0)}_0$, we record an event.

\item Immediately after $T_1$, we re-start the random exponential timer with a new time $V^{(1)}_0$, and also start a new timer labelled $V^{(1)}_1$ whose time is random and distributed according to $F$. We wait until either of the two timers elapses, that is we record both the wait time $W_1 := \min(V^{(1)}_0, V^{(1)}_1)$, and the event which triggered $W_1$; if the minimum is achieved by the exponential timer $V^{(1)}_0$, we say that the second event at $T_2 := T_1 + W_1$ is a background event; otherwise, we say that $T_2$ was triggered by (or caused by) $T_1$.

\item[3--$\infty$] We iteratively apply step 2 above, so that at time $T_i$, we wind back the exponential timer $V^{(i)}_0$, and start $i$ further timers, each associated with a previous event. The timer $V^{(i)}_j$ associated with the event $0 \leq T_j < T_i$ is defined by
\begin{equation}
V^{(i)}_j \overset{d}{=} V - \left(T_i - T_j\right) \mid V \geq \left(T_i - T_j\right)
\label{eq:sampling_times}
\end{equation} 
where $V$ is distributed according to $F$. This ensures that the timer delay will be greater than the gap between $T_j$ and $T_i$, or in other words, that $V_j^{(i)}$ does not fall within the already-observed period ($\mathcal{H}_t$). 

The minimum $W_i$ of all the $V^{(i)}_j$ is the time until the next event, so we set $T_{i+1} := T_i + W_i$, and the parent of $T_{i+1}$ is set as the corresponding index.
   
\end{enumerate}

In Theorem \ref{thm:hazards} we show that the times constructed in this way constitute a realisation of a Hawkes process with a background rate equal to the exponential rate $\eta$, and excitation given by the hazard function of $F$, $h(t) = f(t)/\left[1 - F(t)\right].$ 


The proof uses the following elementary lemma.
\begin{lemma} 
Let $T_i$, $i = 1\ldots d$, be independent positive random variables, interpreted as waiting times, with densities $f_i(t)$. Let $h_i(t) = f_i(t)/S_i(t)$ be the hazard rate of the $i$-th time. Then the minimum waiting time $T := \min\left(T_1, \ldots, T_d\right)$ has the hazard rate $h(t) = h_1(t) + \cdots + h_d(t)$. 
\label{lem:competing}
\end{lemma}
\begin{proof}
Recall that for a positive random variable $T$, the survival function $S(t) = P(T > t)$ is related to the hazard by $\frac{\partial}{\partial t} \log(S(t)) = - h(t)$. We have
\begin{equation*}
S_T(t) = P(T > t) = P(T_1 > t, \ldots, T_d > t) 
= \Pi_{i=1}^d P(T_i > t)
\end{equation*}
by independence, and then
\begin{equation*}
S_T(t) = \prod_{i=1}^d \exp\left(\int_0^t h_i(s) \opd s\right) 
= \exp\left( \sum_{i=1}^d \int_0^t h_i(s) \opd s\right) 
= \exp\left( \int_0^t \sum_{i=1}^d h_i(s) \opd s \right),
\end{equation*}
implying that the sum of hazards is the hazard associated with $T$.
\end{proof}

\noindent In the situation of Lemma \ref{lem:competing} we will say that the $d$ hazards are competing, and this situation arises in the point process constructed earlier.


\begin{theorem}
Suppose that the process $N$ is constructed as above, where the exponential timer has distribution $\mathrm{Exponential}(\eta)$ and the distribution $F$ admits a hazard function $h(t) = f(t)/(1 - F(t))$. Suppose additionally that $\int_0^\infty h\, \opd t < 1$. Then $N$ has the conditional intensity
$$
\lambda(t) = \eta + \int_{0}^t h(t-s) \, N(\opd s).
$$
\label{thm:hazards}
\end{theorem}
\begin{proof}
Immediately after $T_{i-1}$, the time $W_{i-1}$ until the next event is given by the minimum of $V^{(i-1)}_0$ (background event) and the minimum of the $V^{(i-1)}_j \overset{d}{=} V - \left(T_{i-1} - T_j\right) \mid V \geq \left(T_{i-1} - T_j\right)$. Using Lemma \ref{lem:competing}, we know that the hazard function associated with $W_{i-1}$ is the sum of the hazards for each of these cases. Now the hazard function for $V^{(i-1)}_j$, $j = 1, \ldots, i$ at the time $u \geq 0$ after $T_{i-1}$ is given by
\begin{align*}
h_j(u) :&= \frac{\bP{u \leq V^{(i-1)}_j \leq u + \opd u \mid V^{(i-1)}_j \geq u}}{\opd u} \\
&= \frac{\bP{u \leq V - ( T_{i-1} - T_j ) \leq u + \opd u \mid V - ( T_{i-1} - T_j ) \geq u}}{\opd u} \\
&= \frac{\bP{u + T_{i-1} - T_j \leq V \leq u + T_{i-1} - T_j + \opd u \mid V \geq u + T_{i-1} - T_j}}{\opd u} \\
&= h_V(u + T_{i-1} - T_j),
\end{align*}
or in the ``absolute'' time variable $t = u + T_{i-1}$, this is $h_v(t - T_j)$. In the preceding development, the condition $V \geq (T_{i-1} - T_j)$ ensures that $V^{(i-1)}_j$ is a waiting time, which is necessary for the definition of the hazard function. 

For $j = 0$, the waiting time is exponentially distributed, so that $h_0(t) \equiv \eta$.
Now the total hazard is 
$$
h(t) = \sum_{j=0}^i h_j(t) = \eta + \sum_{j=1}^i h(t - T_j) = \eta + \int_{-\infty}^t h(t-s) \, N(\opd s),
$$
as required.
\end{proof}

We now summarise the above construction as an algorithm:

\begin{center}
\begin{algorithm}[H]
 \KwData{distribution $V$ with hazard $h$ and CDF $F$, end time $T > 0$}
 \KwResult{a sample $(T_i, P_i)$ from the Hawkes process}
 $i := 1$, $t_1 \sim \mathrm{Exp}(\eta)$\;
 \While{$t_i < T$}{
  sample $V^{(i)}_0 \sim \mathrm{Exp}(\eta)$\;
  \For{$j = 1, \ldots, i$}{
  sample $V \sim F$ and set $V^{(i)}_j = V - \left(T_i - T_j\right) \mid V \geq \left(T_i - T_j\right)$\;
  }
  set $W_i = \min_j V^{(i)}_j$ and $P_i = \arg \min V^{(i)}_j$\;
  set $T_{i + 1} = T_i + W_i$.
 }
 \vspace{1em}
 \caption{Simulation algorithm for the univariate Hawkes process.}
 \label{alg:sim}
\end{algorithm}
\end{center}

We note that the conditional simulation step in Algorithm \ref{alg:sim} can be implemented easily if the variates $W$ are generated using the inversion method. In this case, we use the equality in distribution $X = F_X^{-1}(U)$ where $U \sim \mathrm{Unif}(0, 1)$ to simulate $X$. A variant on this, adapted for left-truncated random variables relies on the fact that
$$
X \mid X > a \overset{d}{=} F_X^{-1}(U) \text{ where } U \sim \mathrm{Unif}(F_X(a), 1).
$$

\section{Conclusion}
\label{sec:conclusion}

In this paper, we sought to offer a new perspective on the Hawkes process by interpreting the excitation function $h$ as a hazard function. This perspective sheds more light on existing parametric models, suggests a straightforward simulation algorithm, and also suggests avenues for further research.


We have only mentioned a few parametric distributions to define new excitation functions in Section \ref{sec:cure_rate}. This comes down to finding families of defective distributions with positive support. Obviously, in view of the remarks in Section \ref{sec:tech_def}, further examples abound by truncating parametric distributions, but we believe there remains interest in studying defective distributions supported on the whole positive half-life. To this end, we mentioned the Marshall-Olkin and Kumaraswamy families, which are a promising source of defective distributions, though the corresponding hazard functions do not always have simple forms; we expect that other families may be developed which maintain simple forms of $h(t)$ but add flexibility too.

A marked Hawkes point process, in which the hazard function depends both on the time delay and the past mark, can be seen in much the same way as an unmarked one. The model of most interest is of course the Cox proportional hazards model \cite{cox_regression_1972}, where a baseline hazard function such as those presented here is multiplied by a covariate-dependent factor. A possible difficulty may lie in the stability criterion of the marked Hawkes process, which depends on the distribution of the marks, and hence, on the values of the covariates.

While we focused interpretation and simulation, we imagine that the correspondence between excitation and hazard functions will also lend itself to estimation of the Hawkes process, whether as a graphical tool to select parametric forms or through existing methods in the estimation of survival models. Our initial efforts in this direction were hampered by the complicated data generating process of the individual wait times $V^{(i)}_j$ (see Section \ref{sec:simulation}), but perhaps existing work on the mixture of survival distributions, whether frequentist or Bayesian, could help elucidate the estimation of $h$.

Finally, our simulation algorithm can be easily extended to the case where the background rate $\eta$ varies with time, by simulating its waiting-time distribution in much the same way. It also extends to the recent and very interesting Renewal-Hawkes model of \cite{wheatley_hawkes_2016}, and one could see their paper as adding a non-exponential hazard function for the background process, associated with the renewal distribution $F$, to the additive decomposition of hazards. Marks could also be simulated either independently or conditional on the past. One could also imagine an extension of our simulation algorithm to the multivariate setting, though this seems more difficult because of the high number of variates to be simulated.

As continuous-time event data continues to become more abundant and profitable, it is clear that more efficient and mature computational tools must be developed for modelling and simulation. In the Python ecosystem, a significant advance has been made by Emmanuel Bacry and collaborators with the Tick project (\texttt{https://x-datainitiative.github.io/tick/}). A dedicated \texttt{R} package with well-tested and optimised methods would also be welcome.

\bibliography{Hawkes}

\end{document}